\documentclass[aps,prl,reprint,amsmath,amssymb,superscriptaddress,showpacs]{revtex4-1}
\usepackage{hyperref,amsthm}

\newcommand{\ket}[1]{\ensuremath{|{#1\rangle}}} 

\newcommand{\braket}[2]{\ensuremath{{\langle #1}|{#2 \rangle}}}
\newcommand{\ketbra}[2]{\ensuremath{|{#1 \rangle}{\langle #2}|}}

\newcommand{\mmt}[1]{#1}
\newcommand{\resp}{\text{pr}}
\newcommand{\born}{\text{Pr}}
\newcommand{\D}{\text{d}}
\providecommand{\abs}[1]{\left\lvert#1\right\rvert}

\newtheorem*{definition}{Definition}
\newtheorem*{assumption}{Assumption}
\newtheorem*{theorem}{Theorem}
\newtheorem*{lemma}{Lemma}

\begin{document}

\title{No-Go Theorem for the Composition of Quantum Systems} 

\author{Maximilian Schlosshauer}

\affiliation{Department of Physics, University of Portland, 5000 North Willamette Boulevard, Portland, Oregon 97203, USA}

\author{Arthur Fine}

\affiliation{Department of Philosophy, University of Washington, Box 353350, Seattle, Washington 98195, USA}


\begin{abstract}
Building on the Pusey--Barrett--Rudolph theorem, we derive a no-go theorem for a vast class of deterministic hidden-variables theories, including those
consistent on their targeted domain. The strength of this result throws doubt on seemingly natural assumptions (like the ``preparation independence''
of the Pusey--Barrett--Rudolph theorem) about how ``real states'' of subsystems compose for joint systems in nonentangled states. This points to constraints in modeling tensor-product states, similar to constraints demonstrated for more complex states by the Bell and Bell--Kochen--Specker theorems.
\end{abstract}

\maketitle

Studies by Pusey, Barrett, and Rudolph (PBR) \cite{Pusey:2012:np} and others \cite{Colbeck:2012:cr,Hardy:2012:rr} demonstrate a no-go theorem for properties of ontological \cite{Harrigan:2010:pl} hidden-variables models. We show that if the strategy of the demonstration is viable, it leads to a theorem like that claimed by von Neumann generations ago \cite{Neumann:1932:gq}; that is, it leads to a broad no-go theorem for deterministic hidden-variables models, including successful models known to reproduce the quantum statistics for the systems in question \cite{Kochen:1967:hu,Bell:1966:ph,Lewis:2012:qs}. This startling consequence calls for an examination of the elements essential to the strategies underlying these no-go theorems. One critical element is an assumption of how hidden variables of component systems relate to hidden variables of the composite in product states. We show that the physical rationale for composition principles of this kind overreaches. Our results, and those of PBR, highlight that the tensor-product structure required for composite systems, even for those prepared in nonentangled states, can open up new possibilities that cannot be accommodated by ontological hidden-variables models that embody classical intuitions about how hidden variables (``real states'') of quantum systems ought to compose.

\emph{Hidden variables.---}In the models under consideration, each quantum state $\ket{\psi}$ in the state space of a given system is associated with a nonempty set $\Lambda_\psi$ that supports a probability density function $p_\psi(\lambda)>0$ for $\lambda \in \Lambda_\psi$, where $\int_{\Lambda_\psi} p_\psi(\lambda) \, \D\lambda = 1$. We will refer to a complete state $\lambda$ as associated with $\ket{\psi}$ if $\lambda \in \Lambda_\psi$. With probability $p_\psi(\lambda)$, preparing $\ket{\psi}$ results in a $\lambda$ associated with $\ket{\psi}$. In general, different systems, each prepared in $\ket{\psi}$, may have different $\Lambda_\psi$ and $p_\psi(\lambda)$. 

Adopting Einstein's language for quantum incompleteness \cite{Einstein:1936:ji}, PBR call the hidden variables ``physical states'' or ``real physical states'' \cite{Pusey:2012:np}. That terminology signals that the $\lambda$s are regarded as representing real aspects of a quantum system. But since the structure of a hidden-variables model cannot actually fix the nature (or reference) of the $\lambda$s, we will use the neutral language often used by Bell \cite{Bell:1987:su} of ``complete states,'' except where realist intuitions come into play. 

The states are complete in the sense that they suffice to determine the probable responses to measurements of any observable $\mmt{M}$ defined on the state space of the system. Here and below, we assume that $\mmt{M}$ is discrete. Then we have a response function $\resp(\mmt{M}=k \mid \lambda)$ that, given a system in complete state $\lambda$, yields the probability that a measurement of $\mmt{M}$ results in eigenvalue $k$. Two elements characterize the ontological framework \cite{Harrigan:2010:pl} employed by PBR. (i) Response functions do not depend on the quantum state (unless that dependence is written into the $\lambda$s). (ii) Hidden variables do not play any role in accounting for measurement inefficiencies, so that
\begin{equation}\label{eq:unity}
\sum_{k\in S(\mmt{M})} \resp(\mmt{M}=k \mid \lambda) = 1,
\end{equation}
where $S(\mmt{M})$ is the spectrum of $\mmt{M}$. We have shown elsewhere \cite{Schlosshauer:2012:pr} that the PBR theorem requires both (i) and (ii). The Born probability $\born(\mmt{M}=k \mid \ket{\psi})$ that a measurement of $\mmt{M}$ in state $\ket{\psi}$ results in eigenvalue $k$ is obtained from
\begin{equation}\label{eq:matchborn2}
\born(\mmt{M}=k \mid \ket{\psi}) = \int_{\Lambda_\psi} \resp(\mmt{M}=k \mid \lambda)\, p_\psi(\lambda) \, \D\lambda.
\end{equation}
In the general case, this setup leaves open the possibility that the ``reality'' represented by $\lambda$ is the quantum state itself. In the deterministic case, where all response functions yield probability 0 or 1, this cannot happen.

\emph{The PBR example.---}We say that states $\ket{\psi}$ and $\ket{\phi}$ overlap just in case $\Lambda_\psi \cap \Lambda_\phi$ is nonempty. In the example PBR use to illustrate their theorem \cite{Pusey:2012:np}, they consider two quantum systems, each independently prepared in either $\ket{1}$ or $\ket{2}$, where $\abs{\braket{1}{2}}=2^{-1/2}$. Suppose the states overlap. Then there is some nonzero probability that the preparations result in $\lambda$s each associated with both  $\ket{1}$ and $\ket{2}$. States $\ket{1}$ and $\ket{2}$ span a two-dimensional space $H_0$. Consider $H=H_0\otimes H_0$, which contains the product states $\ket{x,y} \equiv \ket{x}\otimes\ket{y}$, $x,y=1,2$. Using Bell states, PBR display an orthonormal basis $\{ \ket{\xi_{xy}} \}$ of $H$ such that $\braket{\xi_{xy}}{x,y} = 0$. Then for any maximal measurement $\mmt{M}$ with eigenstates $\ket{\xi_{xy}}=\ket{k_{xy}}$ (where $k_{xy}$ is the corresponding eigenvalue), 
\begin{equation}\label{eq:yh}
\born(\mmt{M}=k_{xy} \mid \ket{x,y}) = \abs{\braket{\xi_{xy}}{x,y}}^2 = 0.
\end{equation}
If the pair of $\lambda$s associated with both $\ket{1}$ and $\ket{2}$ constituted a hidden variable $\lambda_c$ associated with all four product states $\ket{x,y}$, then the response function for $\lambda_c$ would contribute to the Born probabilities in Eq.~\eqref{eq:yh} for all four states $\ket{x,y}$ simultaneously. For such a  $\lambda_c$, Eqs.~\eqref{eq:matchborn2} and \eqref{eq:yh} imply that
\begin{equation}\label{eq:sh}
\resp(\mmt{M}=k_{xy} \mid \lambda_c) = 0 \quad \text{for $x,y=1,2$}. 
\end{equation}
Thus, the $\mmt{M}$ measurement would have no outcome, contradicting Eq.~\eqref{eq:unity}.

The demonstration of a violation of Eq.~\eqref{eq:unity} does not use the full probability rule of Eq.~\eqref{eq:matchborn2}. All one uses is that where the Born probabilities say ``no'' to a measurement outcome, as in Eq.~\eqref{eq:yh}, the appropriate response function also says ``no,'' as in Eq.~\eqref{eq:sh}. This motivates the following definition.

\begin{definition}
\emph{(Tracking)} Consider a system $S$ with state space $H$. A hidden variable $\lambda$ \emph{tracks} $\ket{\psi} \in H$ on $S$ if and only if, for all observables $\mmt{M}$ on $S$, whenever $\born(\mmt{M}=k \mid \ket{\psi})=0$, then $\resp(\mmt{M}=k \mid \lambda)=0$ \footnote{Tracking is similar to  ``possibilistic completeness,'' assumed in a related no-go theorem \cite{Hardy:2012:rr}.}.
\end{definition}

Thus $\lambda$ tracks $\ket{\psi}$ if and only if whenever the outcome probabilities assigned by $\lambda$ are nonzero, the Born outcome probabilities for a system prepared in $\ket{\psi}$ are nonzero. Equation~\eqref{eq:matchborn2} implies that if $\lambda$ is associated with $\ket{\psi}$ on $S$, then $\lambda$ tracks $\ket{\psi}$ on $S$. Thus, tracking $\ket{\psi}$ is a necessary condition for association with $\ket{\psi}$. The converse is not true.

\emph{Composition.---}Because complete states are complete only for measurements on a given state space, deriving Eq.~\eqref{eq:sh} from Eq.~\eqref{eq:yh} requires assumptions about complete states for composites beyond what is built into the hidden-variables structure so far. Thus PBR introduce an assumption they call preparation independence: ``systems \dots\ prepared independently have independent physical states'' (see p.~475 of Ref.~\cite{Pusey:2012:np}). 

The ``independence'' referred to here is twofold. One aspect encompasses stochastic independence ($PI_{st}$) of the $\lambda$s that result from preparing the quantum states. More importantly, to derive a violation of Eq.~\eqref{eq:unity}, preparation independence must encompass a composition principle, $PI_c$, that allows those $\lambda$s to function independently of the quantum states actually prepared. We could capture this compositional aspect of independence by assuming that if $\lambda_1$ is associated with $\ket{1}$ of system $S_1$ and $\lambda_2$ is associated with $\ket{2}$ of $S_2$, then ($\lambda_1$, $\lambda_2$) constitutes a complete state associated with $\ket{1}\otimes \ket{2}$ for the composite system formed from $S_1$ and $S_2$. In fact, the following weaker assumption suffices.

\begin{definition}
\emph{($PI_c$)} If $\lambda_1$ is associated with $\ket{1}$ of system $S_1$ and $\lambda_2$ is associated with $\ket{2}$ of system $S_2$, then the pair ($\lambda_1$, $\lambda_2$) tracks $\ket{1}\otimes \ket{2}$ on the composite system formed from $S_1$ and $S_2$ \footnote{Our argument could allow more general functions defined on $(\lambda_1, \lambda_2)$, provided those functions do not depend on specific preparation procedures, measurements, or quantum states.}.
\end{definition}    

According to $PI_c$, although each $\lambda$ is associated with some pure state, these need not be pure states actually prepared on a given occasion. Thus, suppose two systems are prepared independently---say, one in $\ket{1}$ and the other in $\ket{2}$---resulting (respectively) in complete states $\lambda_1$ and $\lambda_2$. Then $PI_c$ implies that ($\lambda_1$, $\lambda_2$) tracks $\ket{1}\otimes \ket{2}$ on the composite system. But $PI_c$ also implies (counterfactually) that had different states $\ket{\alpha_1}$ and $\ket{\alpha_2}$ been prepared with which (respectively) complete states $\lambda_1$ and $\lambda_2$ are also associated, then the same pair ($\lambda_1$, $\lambda_2$) would simultaneously track $\ket{\alpha_1}\otimes \ket{\alpha_2}$ and $\ket{1}\otimes \ket{2}$ on the composite (even had the hypothetical preparations turned out different $\lambda$s).

When applied to the PBR example where $\abs{\braket{1}{2}}=2^{-1/2}$, the counterfactuals supported by $PI_c$ yield the composition rule PBR employ to move from Eq.~\eqref{eq:yh} to Eq.~\eqref{eq:sh}. This is so because $PI_c$ implies that if independent preparations of two systems, each in either $\ket{1}$ or $\ket{2}$, result in complete states associated with both $\ket{1}$ and $\ket{2}$, then there is a $\lambda_c$ that simultaneously tracks all four states $\ket{x,y} \equiv \ket{x}\otimes\ket{y}$, $x,y=1,2$, on the composite formed from $S_1$ and $S_2$.

To cover the general case $0 < \abs{\braket{1}{2}}^2 <1$ developed by PBR, we can extend $PI_c$ to apply to arbitrary tensor products $\ket{x_1}\otimes\ket{x_2}\otimes \ket{x_3}\otimes \cdots$, $x_i \in \{1,2\}$. This results in the compactness principle, tailored to tracking, that we formulated elsewhere \cite{Schlosshauer:2012:pr} as a composition rule sufficient for the PBR argument. While we may add $PI_\text{st}$, it is not needed to derive a violation of Eq.~\eqref{eq:unity}.

PBR do not discuss the physical rationale for assuming $PI_c$, which may seem natural from a realist point of view, where the $\lambda$s associated with $\ket{1}$ and $\ket{2}$ represent all the hard facts relevant to probable measurement outcomes on the respective systems. In the state $\ket{1}\otimes \ket{2}$, the subsystems are not entangled and not interacting (at least not in an entangling manner), and, hence, one might regard their composite as not generating any new facts. So facts about probable outcomes on the two subsystems, taken conjointly, should constitute all the facts about likely outcomes (i.e., about tracking) on the composite system in the product state. Below we develop a challenge to this rationale.

\emph{Tracking.---}We now show that because the PBR strategy requires only tracking rather than association, it is not specific to models with overlap (``epistemic'' models \cite{Harrigan:2010:pl}) but also targets nonoverlapping (``ontic'') models. First, we modify the antecedent (``if'' clause) of $PI_c$ to produce a version purely phrased in terms of tracking.

\begin{definition}
\emph{($PI_{c,tr}$)} If $\lambda_1$ tracks $\ket{1}$ on system $S_1$ and $\lambda_2$ tracks $\ket{2}$ on system $S_2$, then the pair ($\lambda_1$, $\lambda_2$) tracks the product state $\ket{1}\otimes \ket{2}$ on the composite system formed from $S_1$ and $S_2$.
\end{definition}

Since association implies tracking, and since the conclusions of $PI_c$ and $PI_{c,tr}$ are identical, it follows that if $PI_{c,tr}$ and the antecedent of $PI_c$ hold, then $PI_c$ also holds. Thus $PI_{c,tr}$ implies $PI_c$, and $PI_{c,tr}$ is sufficient to generate the PBR contradiction. 

What is the physical rationale for assuming $PI_{c,tr}$? As in $PI_c$, we can think that the $\lambda$s associated with $\ket{1}$ and $\ket{2}$ represent all the hard facts relevant to measurement outcomes with nonzero probability (tracking) on the respective systems. Forming the composite described by $\ket{1}\otimes \ket{2}$ should not generate new facts about outcomes, for the same reasons as in the case of $PI_c$. Hence, all the facts about outcomes that have nonzero probability (tracking) on the components, taken together, should be sufficient to account for outcomes with nonzero probability on the composite. Thus $PI_c$ and $PI_{c,tr}$ have the same rationale.

Since the antecedent of $PI_{c,tr}$ is weaker than the antecedent of $PI_c$, it is more easily satisfied. Thus, as we will now see, $PI_{c,tr}$ opens up the possibility of no-go results broader than those of PBR.

\emph{Deterministic models.---}In deterministic hidden-variables models, all probabilities given by the response functions are 0 or 1. Thus, we write $M(\lambda)$ to denote the eigenvalue $k$ that obtains if $\mmt{M}$ is measured. The Bell--Kochen--Specker (BKS) theorem \cite{Bell:1966:ph,Kochen:1967:hu} targets such models where the state space has dimension $\ge 3$. Essential to that theorem is the rule that an eigenvalue $k$ is assigned to an observable $\mmt{M}$ if and only if the spectral projector $\mmt{P}_k$ belonging to $k$ takes the value 1. (This is equivalent to the function rule assumed in Ref.~\cite{Kochen:1967:hu}, or the additivity of values for commuting operators assumed by Bell \cite{Fine:1978:oo,Mermin:1993:tt}.) The rule mirrors the connection $\born(\mmt{M}=k \mid \ket{\psi}) =\born(\mmt{P}_k=1 \mid \ket{\psi})$ built into the Born probabilities. In two dimensions it is harmless, although in certain higher dimensions we have shown that it falls to compactness \cite{Schlosshauer:2012:pr}. Here, we weaken the rule and consider deterministic hidden-variables theories that are only required to follow it in one direction:
\begin{assumption}
\emph{($A$)} For any state $\ket{\psi}$: if $\lambda \in \Lambda_\psi$ and $\mmt{P}_k(\lambda)=0$, then $M(\lambda)\not=k$, where $\mmt{P}_k$ is the spectral projector onto the $k$ eigenspace of $\mmt{M}$.
\end{assumption}
We now show that for every deterministic hidden-variables theory on a two-dimensional space $H_0$ that satisfies assumption $A$, any two distinct, nonorthogonal quantum states are simultaneously tracked; this means that in such models, the antecedent of $PI_{c,tr}$ is always satisfied.
\begin{lemma}
Suppose $0<\abs{\braket{\psi}{\phi}}^2<1$ and an ontological deterministic hidden-variables theory governs $\lambda$s on the two-dimensional Hilbert space $H_0$ spanned by $\ket{\psi}$ and $\ket{\phi}$. If assumption $A$ holds on $H_0$, then associated with each of these kets $\ket{\psi}$ and $\ket{\phi}$ is a set of measure $\abs{\braket{\psi}{\phi}}^2$ consisting of complete states each of which also tracks the other ket on $H_0$.
\end{lemma}

\begin{proof}
Since every $\lambda \in \Lambda_\psi$ tracks $\ket{\psi}$ on $H_0$, we show that a subset $S \subset \Lambda_\psi$ of these $\lambda$s also tracks $\ket{\phi}$ on $H_0$. Let $\mmt{P}^\perp = \ketbra{\phi^\perp}{\phi^\perp}$ be the projector along the state vector $\ket{\phi^\perp}$ in $H_0$ that is orthogonal to $\ket{\phi}$. Then $\born(\mmt{P}^\perp=1 \mid \ket{\psi}) = 1-\abs{\braket{\psi}{\phi}}^2
\not=1$. Hence, there exist $\lambda \in \Lambda_\psi$ such that $\mmt{P}^\perp(\lambda)=0$; otherwise the overall probability of having $\mmt{P}^\perp=1$ would be 1. With respect to the density $p_\psi(\lambda)$, the set $S$ of such $\lambda$s has measure equal to $\born(\mmt{P}^\perp=0 \mid \ket{\psi})=\abs{\braket{\psi}{\phi}}^2$. Consider any $\mmt{M}$ for which $\born(\mmt{M}=k \mid \ket{\phi})= \born(\mmt{P}_k=1 \mid \ket{\phi})=0$. Then, since the space is two-dimensional, the projector $\mmt{P}_k$ on the $k$ eigenspace of $\mmt{M}$ just projects onto $\ket{\phi^\perp}$; so $\mmt{P}_k=\mmt{P}^\perp$. For any $\lambda \in S$, $\mmt{P}^\perp(\lambda)=0$. Assumption $A$ then implies that $M(\lambda)\not=k$; that is, $\resp(\mmt{M}=k \mid \lambda) = 0$. Thus, every $\lambda \in S$ tracks $\ket{\phi}$ on $H_0$. The same argument applies if we interchange $\ket{\psi}$ and $\ket{\phi}$.
\end{proof}

\begin{theorem}
No ontological deterministic hidden-variables theory satisfying assumption $A$ and the composition principle $PI_{c,tr}$ can reproduce the predictions of quantum mechanics.
\end{theorem}
\begin{proof}
Consider systems $S_1$ and $S_2$, each independently prepared in either $\ket{1}$ or $\ket{2}$, with $\abs{\braket{1}{2}}=2^{-1/2}$. Suppose an ontological deterministic hidden-variables theory satisfying assumption $A$ governs $\lambda$s on the two-dimensional space $H_0$ spanned by $\ket{1}$ and $\ket{2}$. The above lemma establishes that there exists a nonempty set $S$ of $\lambda$s that track both $\ket{1}$ and $\ket{2}$ on both $S_1$ and $S_2$. Since $PI_{c,tr}$ supports the same counterfactuals as $PI_c$, $PI_{c,tr}$ implies that for any $\lambda \in S$, $\lambda_c=(\lambda, \lambda)$ tracks the four states $\ket{x,y}$, $x,y=1,2$, on the composite system represented by $H_0\otimes H_0$. Thus, for this $\lambda_c$ and the PBR measurement $M$, we obtain the PBR contradiction, Eq.~\eqref{eq:sh}. For arbitrary distinct nonorthogonal states $\ket{1}$ and $\ket{2}$, we can extend $PI_{c,tr}$, analogous to $PI_c$, to cover arbitrary tensor products $\ket{x_1}\otimes\ket{x_2}\otimes \ket{x_3}\otimes \cdots$, $x_i \in \{1,2\}$, and then apply PBR's quantum circuit to arrive at a contradiction. 
\end{proof}

\emph{Discussion.---}The no-go theorem derived here is very strong, stronger than the BKS theorem in two respects. First, the assumption $A$ it requires is weaker than the BKS condition. Second, unlike the BKS theorem, it applies to systems with two-dimensional state spaces. It shows that no deterministic qubit model (or submodel) satisfying assumption $A$, even if it is consistent with the quantum predictions on its domain, can be extended via the composition principle $PI_{c,tr}$ to tensor-product states. This applies to almost all the hidden-variables models reviewed in Ref.~\cite{Belinfante:1973:om}, and includes the models for qubit systems of Kochen and Specker \cite{Kochen:1967:hu} and the ontic finite-dimensional model of Bell \cite{Bell:1966:ph,Lewis:2012:qs}, both of which satisfy assumption $A$ and are known to be quantum consistent.
 
The strength of this result calls attention to the operative composition principle $PI_{c,tr}$ and the possibility of flaws in the physical rationale sketched above. Two possibilities stand out. One is the very idea that pairs ($\lambda_1$, $\lambda_2$) resulting from independent preparations suffice alone to determine probable measurement outcomes on the tensor product---as would be the case, for instance, if we could identify the $\lambda$s with the prepared quantum states. That identification, however, is not an option, being incompatible both with determinism and with overlap. Moreover, at least four distinct tensor products are required for a contradiction, whereas only two quantum states are actually prepared. Thus, we need to recognize the possibility that, in addition to ($\lambda_1$, $\lambda_2$), facts about the context of the actual preparations or subsequent measurements may be needed in order to track the product states. (Ignoring such contextual factors is at the root of the BKS theorem.) But if contextual factors need to be taken into account, a composition principle guaranteeing a complete state $\lambda_c$ that simultaneously tracks all the product states, context free, need not hold. 

A second possible flaw arises from the circumstance that measurements like $M$ (and PBR's quantum circuit) are entangling. They engage the tensor-product structure to generate facts pertaining to the composite as a whole. Like correlations, such facts are not accessible from the isolated subsystems. Contrary to the stated rationale, forming composites, even ones described by tensor-product states, makes available new, relational facts about measurement outcomes.

These reservations about $PI_{c,tr}$ could be taken to undermine the no-go theorem developed here. But both apply in exactly the same way to the compositional assumption $PI_c$ of preparation independence required for the PBR theorem. Thus the composition principles $PI_{c,tr}$ and $PI_c$ stand or fall together, depending on whether we credit the physical rationale or the reservations. 

There is an important constructive message here. Recall that other significant no-go theorems, such as the BKS theorem and the Bell theorem, were based on natural assumptions (noncontextuality, locality) supporting counterfactuals in a classical setting. The positive lesson from those no-go theorems was to throw such assumptions into doubt when imported to the quantum world. We suggest the same lesson here. While entanglement and ``quantum nonseparability'' indicate that simple rules of composition for ``real states'' are unlikely, one might have assumed that when modeling a tensor-product state, the compositional aspect of preparation independence, $PI_c$, should be viable. Our results challenge this assumption. They caution against classical, realist intuitions about how ``real states'' ought to compose, even in the absence of entanglement. It would be interesting to investigate the status of composition rules in other classes of hidden-variables models.

\begin{acknowledgments}
We thank J.\ Malley for useful correspondence.
\end{acknowledgments}


%

\end{document}